\documentclass[11pt]{article}
\usepackage{amsmath}
\usepackage{amssymb}
\usepackage{amsfonts}
\usepackage{latexsym}
\newcommand{\mc}{\mathcal}
\newcommand{\mb}{\mathbb}
\newcommand{\D}{\mc D}
\newcommand{\Hil}{\mc H}
\newcommand{\A}{\mc A}
\newcommand{\LL}{\mc L}

\newcommand{\beano}{\begin{eqnarray*}}
\newcommand{\enano}{\end{eqnarray*}}

\newcommand{\LD}{{\LL}^\dagger (\D)}
\newcommand{\dg}{\dagger}
\newcommand{\restr}{\upharpoonright}
\newtheorem{defn}{Definition}[section]
\newtheorem{prop}[defn]{Proposition}

\newtheorem{lemma}[defn]{Lemma}
\newtheorem{cor}[defn]{Corollary}
\newenvironment{proof}{\noindent {\bf Proof --}}{\hfill$\square$ \vspace{3mm}\endtrivlist}

\begin{document}

\thispagestyle{empty}

\vspace*{0.7cm}

\begin{center}
{\Large \bf Algebraic dynamics in O*-algebras:\\ a perturbative
approach}

\vspace{1.5cm}

{\large F. Bagarello }
\vspace{3mm}\\
 Dipartimento di Matematica e Applicazioni \\
Facolt\`a d' Ingegneria - Universit\`a di Palermo \\ Viale delle
Scienze,       \baselineskip15pt
     I-90128 - Palermo - Italy \vspace{2mm}\\

\vspace{3mm} and

\vspace{3mm} {\large C.Trapani}
\vspace{3mm}\\
 Dipartimento di Matematica e Applicazioni \\
Universit\`a di Palermo\\Via Archirafi 34, \baselineskip15pt
I-90123 - Palermo - Italy\vspace{2mm}

\end{center}

\vspace*{1cm}
\begin{abstract}
\noindent In this paper the problem of recovering an algebraic
dynamics in a perturbative approach is discussed. The mathematical
environment in which the physical problem is considered is that of
algebras of unbounded operators endowed with the quasi-uniform
topology. After some remarks on the domain of the perturbation,
conditions are given for the dynamics to exist as the limit of a
net of regularized linear maps.

\end{abstract}

\vfill

\section{Introduction}In the so-called algebraic approach to quantum systems, one of the basic problems to solve consists in the rigorous definition of the algebraic dynamics, i.e. the time evolution of observables and/or states. For instance, in quantum statistical mechanics or in quantum field theory one tries to recover the dynamics by performing a certain limit of the strictly {\em local} dynamics. However, this can be successfully done only for few models and under quite strong topological assumptions (see, for instance, \cite{sakai} and references therein). In many physical models the use of local observables corresponds, roughly speaking, to the introduction  of some {\em cut-off} (and to its successive removal) and this is in a sense a general and frequently used procedure, see \cite{thi,bagarello,lass1} for conservative and \cite{sew,sewbag} for dissipative systems.\\
Introducing a cut-off means that in the description of some
physical system, we know a {\em regularized} hamiltonian $H_L$,
where $L$ is a certain parameter closely depending on the nature
of the system under consideration.
We assume that $H_L$ is a bounded self-adjoint operator in the Hilbert space $\mc H$ of the physical system.\\
There are several possible situations of some interest. Among these we will consider the following ones:\\

\vspace{2mm}
\noindent a) {\em $H_L$ converges to an operator $H$}\\
This is apparently the simplest situation. Of course we should specify the sense in which the convergence is understood. But for the moment, we want only focus on the possible problems that arise. \\
For each fixed $L$, we know the solution of the dynamical problem,
i.e., we know the solution of the Heisenberg equation
\begin{equation}
 i \frac{d\alpha_L^t(A)}{dt}= [H_L,\alpha_L^t(A)].
\label{heis1}
\end{equation}
This solution, $\alpha_L^t(A)=e^{iH_Lt}A e^{-iH_Lt}$, would give
the {\em cut-offed} dynamics of the system. Then it make sense to
ask the question as to whether $\alpha_L^t(A)$ converges, possibly
in the same sense as $H_L$ converges to $H$, to the solution
$\alpha^t(A)$ of the Heisenberg equation
\begin{equation} i \frac{d\alpha^t(A)}{dt}= [H,\alpha^t(A)].
\label{heis2}
\end{equation}
It is worth stressing that even though $H$ is a well defined self-adjoint operator, it is in general, unbounded. For this reason, while the right hand side of Eqn. \eqref{heis1} is perfectly meaningful, the need of clarifying the right hand side of Eqn. \eqref{heis2} is always in order since $H$ is certainly not everywhere defined in $\Hil$.\\
Of course, the analysis of the convergence of $\alpha_L^t (A)$ to
$\alpha^t(A)$ (in this case where $\alpha^t(A)$ is explictly
known) is significant only for deciding the accuracy of the
approximation of $\alpha^t(A)$ with $\alpha_L^t (A)$.

\vspace{2mm}
\noindent b) {\em $H_L$ does not converge}\\
In this case, the situation becomes more difficult and a series of questions arise whose answer is highly non trivial.\\
As a first step, one could begin with considering the derivations
$$\delta_L (A) = i [H_L,A] $$
that give, at infinitesimal level, the dynamics of the system.\\
The first question, of course, is if these derivations converge, in a certain sense, to a derivation $\delta$ and which properties this derivation $\delta$ enjoys. For instance, is it a spatial derivation? (i.e, is there a symmetric operator $H$ that {\em implements}, at least in a generalized sense, the derivation? \cite{camillo}) \\
Further, can this derivation be integrated to some automorphisms group of the operator algebra we are dealing with? Or, conversely, since $\delta_L$ can be integrated without any problem, $\alpha_L^t(A)= {\rm e}^{iH_Lt} A {\rm e}^{-iH_Lt}$, what can be said about the limit of $\alpha_L^t$? And how are these two problems related?\\
These questions are well-known not to admit an easy general
solution.

\vspace{2mm} In this paper we will be mainly concerned with
situation a) above, while we will only make few comments on the
more difficult situation b) which will be considered in more
details in a future paper.

Our basic assumptions is that the hamiltonian $H$ of the system
can be expressed in the form
$$ H = H_0 + B;$$
in other words, our approach is tentatively {\em perturbative}:
indeed, we suppose that we have full knowledge of the {\em
unperturbed system} whose hamiltonian is $H_0$. In other words,
given $H$ we can extract what we call a free hamiltonian $H_0$,
which we know in all details, and
consider $B:=H-H_0$ as a perturbation of $H_0$ itself.\\
As we have already said, handling with unbounded operators poses a problem of domain for the algebra generated by the powers of the hamiltonian $H_0$. The natural choice is to take the set of $C^\infty$-vectors of $H_0$. Once a perturbation $B$ is introduced, it is natural to ask ourselves in which sense the corresponding domain for $H$ is related to that of $H_0$.\\
This is the main subject of Section 2,  where we start with the assumption that $\D^\infty(H)=\D^\infty(H_0)$ and derive some properties of the corresponding {\em quasi-uniform} topologies that the two operators define. \\
Then we give, in a quite general way, conditions on two
self-adjoint operators $H_0$ and $H$ for $\D^\infty(H)$ and
$\D^\infty(H_0)$ to coincide.

In Section 3, we come back to the problem of describing the
dynamics of the perturbed system as limit of a cut-offed dynamics.
In other words, we introduce a regularized hamiltonian $H_L=
Q^0_LHQ^0_L$ where the $Q^0_L$'s are certain spectral projection
of the unperturbed hamiltonian $H_0$ and we look for conditions
under which the unitary group generated by $H_L$  converges to
that generated by $H$. A class of examples fitting our hypotheses
is also given.

The main scope of the paper is to try and construct a mathematical
environment where this kind of problems can be successful treated
and also to develop techniques that could be adapted for the study
of the more relevant case b) outlined above. It is worth stressing
that this is a rather common situation in physics (think of
mean-field models or systems with ultra-violet cutoff \cite{thi,
bagarello}) and a perturbative approach should also be considered
for the derivations that describe the system at infinitesimal
level. A short discussion on this point is made in Section 4.

 \section{The mathematical framework}
We begin this Section with summarizing some known facts on
unbounded operator algebras and their topological properties. We
refer to \cite{lass1, lass2, ctrev, schm} for full details.

Let $\D$ be a dense domain in Hilbert space $\Hil$; with $\LD$ we
denote the set of all weakly continuous endomorphisms of $\D$.
Then to each operator $A \in \LD$ we can associate an operator
$A^\dg\in \LD$ with $A^\dg = A^*_{\restr \D}$
where $A^*$ is the usual Hilbert adjoint of $A$. Then $\LD$, under the usual operations and the involution $^\dg$, is a *-algebra of unbounded operators or, simply, an O*-algebra.\\
Let now $S$ be a selfadjoint operator in $\Hil$ and
 $$\D:=\D^\infty(S)= \bigcap_{n\geq 1}D(S^n).$$
Then $\D$ endowed with the topology $t_S$ of $\D^\infty(S)$
defined by the set of seminorms
$$ f \mapsto \|S^nf\|, \quad n=0,1,\ldots$$
or, equivalently
$$ f \mapsto \|(1+S^{2n})^{1/2}f\|, \quad n=0,1,\ldots $$
is a reflexive Fr\'echet space and the topology $t_S$ is
equivalent to the topology $t_{\LD}$ defined on $\D$ by the set of
seminorms
$$ f \mapsto \|Af\|, \quad A \in \LD.$$

In the *-algebra $\LD$ several topologies can be defined. For the
purposes of this paper we will only need the {\em quasi-uniform
topology} defined on $\LD$ in the following way. Put
$$ \|A\|^{{\mathcal N}, B}=\sup_{\phi \in {\mathcal N}} \|BA\phi\|, \quad B\in \LD,\; {\mathcal N} \mbox{ bounded in } \D[t_{\LD}].$$
Then, the quasi-uniform topology, $\tau_*^{\D}$ on $\LD$ is
defined by the set of seminorms:
$$ A\in \LD  \longrightarrow \max\{ \|A\|^{{\mathcal N}, B},  \|A^\dg\|^{{\mathcal N},
B}\}.$$
In the case where $\D = \D^\infty(S)$, the quasi-uniform topology on $\LD$ can be described in an easier way.\\
Indeed, let ${\mathcal F}$ denote the class of all positive, bounded and continuous functions $f(x)$ on $\mathbb{R}_{+}$, which are decreasing faster than any inverse power of $x$, i.e., $\sup_{x \in \mathbb{R}_{+}}x^kf(x)< \infty, \;\;k=0,1,\ldots$. \\
Then, if we put
$$ {\mathcal S}_f = \{ f(S)\phi;\, \phi \in \D, \|\phi\|=1\}$$ for $ f\in {\mathcal F}$, the family $\{{\mathcal S}_f \}_{f\in {\mathcal F}}$ is a basis for the bounded sets of $\D[t_S]$.\\
In practice this means that, for each $t_S$-bounded set ${\mathcal N}$ in $\D$, there exists an ${\mathcal S}_f$ such that ${\mathcal N} \subset {\mathcal S}_f$.\\
This fact easily implies that the quasi-uniform topology,
$\tau_*^{\D}$ on $\LD$ can be, equivalently,  defined by the set
of seminorms:
\begin{equation}
 \LD \in A \mapsto \|A\|_*^{f,k} = \max\{ \|S^kAf(S)\|,\|f(S)AS^k\|\} \quad f\in {\mathcal F},\, k\in \mathbb{N}\cup \{0\}
\label{QUN}
\end{equation}
where the norm on the right hand side of \eqref{QUN} is the usual
norm in ${\mathcal B}(\Hil)$. The *-algebra $\LD[\tau_*^{\D}]$ is,
in this case, a complete locally convex *-algebra, i.e. the
involution and the right- and left-multiplications are continuous.

\vspace{2mm} \noindent{\bf Remark -- } When estimating seminorms
of type \eqref{QUN} we will often consider only the term
$\|f(S)AS^k\|$; this is exactly what is needed when $A=A^\dg$. In
the general case, any $A \in \LD$  is a linear combination of
symmetric elements and so, as far as only estimates are concerned,
the arguments go usually through.

\vspace{2mm} We can now consider more concrete situations. To
begin with, we consider the simplest possible example in which a
physical system is described by a Hamiltonian $H_0$ that
mathematically is a self-adjoint operator; we assume $H_0\geq1$;
then $H_0$ has a spectral decomposition
$$ H_0 = \int_1^\infty \lambda d E(\lambda).$$
We put, for $L\geq 1$
 \begin{equation}
Q^0_L= \int_1^L d E(\lambda) \label{proj}
\end{equation}
and define the regularized hamiltonian by:
$$H_L = Q^0_L H_0 Q^0_L .$$
Then if $\D =\D^\infty(H_0)$ it turns out that the operators $Q^0_L$ and $H_L$ are bounded operators in ${\mathcal B}(\Hil)$ which belong to $\LD$ (the $Q^0_L$'s are indeed projectors) and they commute with each other and with $H_0$.\\
This makes quite easy to prove the following convergence
properties (in what follows the topology $\tau_*^{\D}$ is that
defined in Eqn. \eqref{QUN} with $S$ replaced by $H_0$):
\begin{itemize}
\item[(c1)] $H_L \to H_0$ with respect to the topology $\tau_*^{\D}$
\item[(c2)] ${\rm e}^{itH_L} \to {\rm e}^{itH_0}$ with respect to the topology $\tau_*^{\D}$
\item[(c3)] For each $A \in \LD$, ${\rm e}^{itH_L}A{\rm e}^{-itH_L}\stackrel{\tau_*^{\D}}{\rightarrow}{\rm e}^{itH_0}A{\rm e}^{-itH_0}$
\end{itemize}
All these statements can be derived from Lemma \ref{2.2} below.

The next step consists in considering a hamiltonian
\begin{equation}H= H_0 + B
\label{HAM}
\end{equation}
where $B$ is regarded as a {\em perturbation} of the operator
$H_0$. We suppose that the cut-off is determined by $H_0$, i.e.,
we assume that
\begin{equation} H_L = Q^0_L (H_0 + B) Q^0_L = H_0Q^0_L + Q^0_L B Q^0_L
\label{cutoff}
\end{equation}
where $Q^0_L$ is defined as in Eqn. \eqref{proj} by the spectral family $E(\cdot)$ of $H_0$. The r.h.s. is well defined since $Q^0_LAQ^0_L$ is bounded for any $A \in \LD$. \\
Clearly \eqref{cutoff} must be read as a formal expression unless
the domains of the involved operators are specified. To be more
definite, we make the following assumptions:
\begin{itemize}
\item[(a)] $\D =\D^\infty(H_0)$
\item[(b)] $D(H_0)\subseteq D(B)$ and $H=H_0+B$ is self-adjoint on $D(H_0)$
\item[(c)] $\D^\infty(H_0) = \D^\infty(H)$
\end{itemize}
Under these assumptions, we have:
\begin{lemma} {\ }\\
(1) The topologies $t_{H_0}$ and $t_H$ are equivalent on $\D$;\\
(2) the topologies on $\LD$ defined respectively by the set of
seminorms
$$\LD \in A \mapsto  \max\{ \|H_0^kAf(H_0)\|,\|f(H_0)AH_0^k\|\} \quad f\in {\mathcal F},\, k\in \mathbb{N}$$
and
$$\LD \in A \mapsto  \max\{ \|H^kAf(H)\|,\|f(H)AH^k\|\} \quad f\in {\mathcal F},\, k\in \mathbb{N}$$
are equivalent \label{topol}
\end{lemma}
\begin{proof} The statement (1) follows by taking into account that $H$ is continuous with respect to $t_{H_0}$ and $H_0$ is continuous with respect to $t_H$, according to the fact that the domain is reflexive.\\
The statement (2) follows from (1), since the family of
$t_H$-bounded subsets of $\D$ and the family of $t_{H_0}$-bounded
subsets coincide.
\end{proof}
By the previous Lemma, the topology $\tau_*^{\D}$, can be
described, following the convenience, via the seminorms in $H$ or
by those in $H_0$. Now, we can  prove the following
\begin{lemma} For each $X \in \LD$, $X= \tau_*^{\D}-\lim_{L\to \infty} Q^0_L X Q^0_L$
\label{2.2}
\end{lemma}
\begin{proof} First, notice that, for $\ell \in {\mb N}^+$, we have
$$ \|H_0^{-\ell}(I-Q^0_L)\phi\|^2 = \int_L^\infty \frac{1}{\lambda^{2\ell}}d(E(\lambda)\phi, \phi) \leq \frac{1}{L^{2\ell}}\|\phi\|^2, \quad \forall \phi \in \D $$
and so
$$ \|H_0^{-\ell}(I-Q^0_L)\| \to 0 \mbox{ as } L \to \infty .$$
Let now $f \in {\mc F}$ and $k \in {\mb N}$; then we have:
\begin{eqnarray*}
\lefteqn{\|f(H_0)(B-Q^0_L B Q^0_L)H_0^k\|}\\ &\leq &\|f(H_0)B
H_0^k(I-Q^0_L)\| + \|f(H_0)(1-Q^0_L) B H_0^kQ^0_L\|
\\
&=& \sup_{\|\phi\|=\|\psi\|=1}|<H_0^{-\ell}(1-Q^0_L)\phi, H_0^{k+\ell}B^+f(H_0)\psi>| \\&+& \sup_{\|\phi\|=\|\psi\|=1}|<f(H_0)H_0^\ell B H_0^k Q^0_L \phi, H_0^{-\ell}(1-Q^0_L)\psi>|\\
&\leq& \|H_0^{-\ell}(1-Q^0_L)\| \| H_0^{k+\ell}B^+f(H_0)\| +
\|H_0^\ell f(H_0)BH_0^k\| \|H_0^{-\ell}(1-Q^0_L)\| \to 0
\end{eqnarray*}
for $L \to \infty$.
\end{proof}
Incidentally, this lemma gives a proof of (c1) and (c2) above.
 The proof of (c3) requires the
use of a triangular inequality, of (c2) and of the commutation
rule $[H_0,H_L]=0$.

Taking into account the separate continuity of the multiplication
and the previous lemma, we have:
\begin{cor} $\delta_L(A):= i[A,H_L]$ converges to $\delta(A):= i[A,H]$ with respect to the topology $\tau_*^{\D}$.
\label{deriv}
\end{cor}

Going back to our assumptions on the domains, it is apparent that
conditions (b) and (c) given above are quite strong. It is natural
to ask the question under which conditions on $B$ they are indeed
satisfied.
\subsection{The domain}
Our starting point is an operator
$$ H=H_0 +B$$
under the assumption that the {\em perturbation} $B$ is a
symmetric operator and $D(B) \supseteq D(H_0)$. In general $H$ may
fail to be self-adjoint, unless $B$ is $H_0$-bounded in the sense
that there exist two real numbers $a,b$ such that
\begin{equation}
 \|B\phi\| \leq a \|H_0\phi\|+ b \|\phi\|, \quad \forall \phi \in D(H_0) .
\label{kato}
\end{equation}
If the {\em inf} of the numbers $a$ for which \eqref{kato} holds (the so called {\em relative bound}) is smaller than $1$, then the Kato-Rellich theorem \cite{rs2} states that $H$ is self-adjoint and essentially self-adjoint on any core of $H_0$.\\
This is clearly always true if $B$ is bounded: in this case the
relative bound is $0$.
In conclusion, the Kato-Rellich theorem gives a sufficient condition for (b) to be satisfied.\\
Let us now focus our attention on condition (c). We first discuss
some examples.

\noindent{\bf Example 1 --}To begin with, we stress the fact that
the conditions of the Kato-Rellich theorem are not sufficient to
imply that $\D^\infty(H)= \D^\infty(H_0)$. \\This can be seen
explicitly with a simple example.  Indeed, let us consider the
case where $B= P_f$ with $f \in \Hil \setminus D(H_0)$ and $P_f$
the projection onto the one-dimensional subspace generated by $f$.
It is quite simple to prove that, in this case:
$$ D((H_0+P_f)^2) \cap D(H_0^2)= D(H_0) \cap \{f\}^\perp. $$
This equality implies that neither $\D^\infty(H_0+P_f)$ is a
subset of $\D^\infty(H_0)$ nor the contrary. So, in this example,
$\D^\infty(H)$ and $\D^\infty(H_0)$ do not compare.

\noindent{\bf Example 2--} Let $p$ and $q$ be the operators in
$L^2({\mb R})$ defined by:
$$ (pf)(x)=i f'(x), \quad f \in W^{1,2}({\mb R})$$
$$(qf)(x) = xf(x) , \quad f \in {\mc F}W^{1,2}({\mb R})$$
where ${\mc F}$ denotes the Fourier transform. Let us consider
$$ H_0 = p^2+q^2 $$
then, as is known, $H_0$ is an essentially self-adjoint operator on ${\mc S}(\mb R)$ and this domain is exactly $\D^\infty(H_0)$.\\
Let us now take as $B$ the operator $-q^2$, then
$$ \D^\infty (H) = \{ f \in C^\infty ({\mb R)}): f^{(k)} \in L^2(\mb R), \, \forall k \in {\mb N}\}.$$
Thus, in this case $\D^\infty (H) \supset \D^\infty (H_0)$.\\

\vspace{1mm} In order to construct an example where the opposite
inclusion hold, we start by taking $H_0= p^2$ and $B=q^2$. In this
case,
$$ \D^\infty (H) =  {\mc S}(\mb R) \subset \{ f \in C^\infty ({\mb R)}): f^{(k)} \in L^2(\mb R), \, \forall k \in {\mb N}\}= \D^\infty (H_0).$$
These examples show that all situations are possible, when
comparing $\D^\infty (H)$ and $\D^\infty (H_0)$.

\vspace{2mm} For shortness, we will call $B$ a {\em
KR-perturbation} if it satisfies the assumption of the
Kato-Rellich theorem. Before going forth, we give the following
\begin{prop} Let $A$ and $B$ two selfadjoint operators in Hilbert space $\Hil$. Then
$$ \D^\infty (A) = \D^\infty (B) $$
if, and only if, the following two conditions hold:
\begin{itemize}
\item[(i)] for each $ k \in{\mb N}$ there exists $\ell \in {\mb N}$ such that $D(B^\ell) \subseteq D(A^k)$;
\item[(ii)] for each $ h \in{\mb N}$ there exists $m \in {\mb N}$ such that $D(A^m) \subseteq D(B^h)$.
\end{itemize}
\label{ppp}
\end{prop}
\begin{proof} We put $\D = \D^\infty (A) = \D^\infty (B)$. Because of Lemma \ref{topol}, the topologies $t_A$ and $t_B$ are equivalent. Without loss of generality we assume that $A \geq 0$, $B\geq 0$; this makes the usual families of seminorms defining the two topologies directed. This implies that for each $ k \in{\mb N}$ there exist $\ell \in {\mb N}$ and $C_k > 0$:
$$ \|A^k \phi \| \leq C_k \|B^\ell \phi\|, \quad \phi \in \D.$$
But $\D$ is a core for any power of $B$, therefore for each $f \in
D(B^\ell)$ there exists a sequence $(f_n)$ of elements of $\D$
such that $f_n \to f$ and $(B^\ell f_n)$ is convergent. Then we
get
$$ \|A^k (f_n - f_m) \| \leq C_k \|B^\ell (f_n - f_m) \|\to 0$$
and therefore $f \in D\left( \overline{A^k _{\upharpoonright D}}\right)= D(A^k)$.\\
The proof of (ii) is similar.\\
Let us now assume that (i) and (ii) hold. For any $k \in {\mb N}$
we put
$$ \ell_k = \min \{\ell \in {\mb N}: D(B^\ell) \subset D(A^k)\}.$$
Then we have:
$$ \D^\infty (B) \subseteq \bigcap_{k=1}^\infty D(B^{\ell_k}) \subset \bigcap_{k=1}^\infty D(A^k) = \D^\infty(A).$$
In similar way the converse inclusion can be proven.
 \end{proof}
{\bf Example --} The previous proposition easily implies the
following well-known fact:
$$ \D^\infty (A^k) = \D^\infty (A), \quad \forall k \in {\mb N}.$$
since, (i) and (ii) hold, as is readily seen.
\begin{prop}
Let $A\geq 1$ and $B\geq 1$. Then if
$$ \D^\infty (A) = \D^\infty (B) $$
the following two conditions hold:
\begin{itemize}
\item[(i)] for each $ k \in{\mb N}$ there exist $\ell \in {\mb N}$ such that $A^kB^{-\ell}$ is bounded;
\item[(ii)] for each $ h \in{\mb N}$ there exist $m \in {\mb N}$ such that $B^hA^{-m}$ is bounded.
\end{itemize}
Conversely, if $D^\infty (A) \cap D^\infty (B)$ contains a common
core $\D_0$ for all powers of $A$ and $B$ and both (i) and (ii)
hold, then
$$ \D^\infty (A) = \D^\infty (B) $$
\end{prop}
\begin{proof} Assume that $\D^\infty (A) = \D^\infty (B)=:\D$. As seen in the proof of Proposition \ref{ppp}, the equivalence of the topologies implies that for each $ k \in{\mb N}$ there exist $\ell \in {\mb N}$ and $C_k > 0$:
$$ \|A^k \phi \| \leq C_k \|B^\ell \phi\|, \quad \phi \in \D.$$
which can be written as
$$ \|A^kB^{-\ell} \phi\| \leq C_k \|\phi \|, \quad \phi \in \D.$$
The second condition can be proved in similar way.\\
Now, suppose that $\D_0$ is a common core for all powers of $A$
and $B$ and that (i) and (ii) hold. Then from (i) one gets that
for each $ k \in{\mb N}$ there exist $\ell \in {\mb N}$ and $C_k >
0$:
$$ \|A^k \phi \| \leq C_k \|B^\ell \phi\|, \quad \phi \in \D_0.$$
Proceeding as in the proof of Proposition \ref{ppp} one proves
that for these $k$ and $\ell$, $D(B^\ell) \subseteq D(A^k)$.
Analogously, condition (ii) implies (ii) of Proposition \ref{ppp}.

\end{proof}
\begin{prop} Let $B$ be a KR-perturbation and assume $B:\D^\infty(H_0)\to \D^\infty(H_0)$. Then $\D^\infty(H_0)\subseteq\D^\infty(H)$.
Moreover, if the families of seminorms are directed,
$$ \forall k, s \in {\mb N}\;\; \exists \ell \in {\mb N}, C_k>0: \|H_0^sH^k\phi\|\leq C_k \|H_0^\ell \phi\|, \quad \forall \phi \in \D^\infty(H_0).$$
\label{ggg}
\end{prop}
\begin{proof} In this case, $\D^\infty(H_0)$ is left invariant by $H$; but $\D^\infty(H)$ is the largest domain with this property. Therefore $\D^\infty(H_0)\subseteq\D^\infty(H)$.\\
The given inequality follows easily from the continuity of $H^k$
in $\D^\infty(H_0)$.
\end{proof}

\noindent{\bf Remark --} The above inequality also says that
$t_{H_0}$ is, in general, finer than $t_H$.

\vspace{2mm} \noindent In order to get the equality of the two
domains some stronger condition on $B$ must be added. We have,
indeed:
\begin{prop} Let $B$ be a perturbation of $H_0$ such that $H:=H_0+B$ is selfadjoint on $D(H)=D(H_0)$. In order that
$$ \D^\infty(H) = \D^\infty(H_0)$$
it is necessary and sufficient that the following conditions hold:
\begin{itemize}
\item[(i)] $B: \D^\infty(H_0) \to \D^\infty(H_0)$;
\item[(ii)] $H$ is essentially self-adjoint in $\D^\infty(H_0)$;
\item[(iii)] the topologies $t_{H_0}$ and $t_{H}$ are equivalent on $\D^\infty(H_0)$
\end{itemize}
\label{CNS}
\end{prop}
\begin{proof} The necessity of (i) is obvious. As for (ii), it is well-known that $\D^\infty(H)$ is a core for $H$ (and for all its powers). The necessity of (iii) follows from (1) in Lemma \ref{topol}.\\
We now prove the sufficiency.\\
First, by Proposition \ref{ggg} and (i) it follows that
$\D^\infty(H_0) \subseteq \D^\infty(H)$ and since $H$ is
essentially self-adjoint in $\D^\infty(H_0)$,
$$ \D^\infty \left( \overline{H_{\upharpoonright \D^\infty(H_0)}}\right) = \D^\infty(H)  .$$
But as is well known, the domain on the left hand side is the
completion of $\D^\infty(H_0)$ in the topology $t_H$. The
equivalence of $t_H$ and $t_{H_0}$, in turn implies that
$\D^\infty(H_0)$ is complete under $t_H$ and so the statement is
proved.
\end{proof}

\noindent{\bf Remark --} If $B$ is bounded, then $H= H_0+B$ is
automatically essentially self-adjoint in $\D^\infty(H_0)$

\vspace{3mm} \noindent{\bf Example --}
Let $H_0= p^2 +q^2$; then $\D^\infty (H_0) = {\mc S}(\mb R)$. Let $B= \alpha q$ with $\alpha \in {\mb R}$.\\
Then it is easily seen that $H = p^2 +q^2 + \alpha q$ leaves ${\mc S}(\mb R)$ invariant.\\
Since $H= p^2 + (q- \beta)^2 + \beta^2$ with $\beta = \alpha/2$, it is clear that ${\mc S}(\mb R)$ is a domain of essential self-adjointness for $H$.\\
The equivalence of the topologies $t_{H_0}$ and $t_H$ can be
proven with easy estimates of the respective seminorms. Thus
Proposition \ref{CNS}, leads us to conclude that $ \D^\infty (H_0)
= \D^\infty (H) $.

As a consequence of Proposition \ref{CNS}, we consider now the special case of a perturbation {\em weakly} commuting with $H_0$.\\
Let ${\mc L}^\dagger(\D,\Hil)$ denote the space of all closable
operators $A$ in $\Hil$ such that $D(A)=\D$, $D(A^*)\subset \D$.
As for ${\mc L}^\dagger(\D)$, we put $A^\dg = A^*_{\restr \D}$.

Now, if $\A$ is a $\dagger$-invariant subset of ${\mc
L}^\dagger(\D,\Hil)$, the {\em weak unbounded commutant}
$\A'_{\sigma}$ of $\A$ is defined as
$$ \A'_{\sigma}= \{ Y \in {\mc L}^\dagger(\D,\Hil): <Xf,Y^\dagger g>=<Yf,X^\dagger g>, \forall f,g \in D;\, \forall X \in \A\}.$$
If $T$ is a self-adjoint operator in $\Hil$, we can consider the
O*-algebra ${\mc P}(T)$ generated by $T$ on $\D^\infty(T)$. It is
well-known \cite{schm} that ${\mc P}(T)'_\sigma = \{T\}'_\sigma$.
Furthermore, any $Y \in  \{T\}'_\sigma$ leaves $\D^\infty(T)$
invariant. We now apply these facts to our situation:
\begin{cor} Let $B$ be a perturbation of $H_0$. Assume that $B$ satisfies the conditions:
\begin{itemize}
\item[(i)]$ <H_0f, Bg> = <Bf,H_0g>, \forall f,g \in \D^\infty(H_0)$
\item[(ii)] $H$ is essentially self-adjoint in $\D^\infty(H_0)$
\item[(iii)] $\|H_0f\| \leq \|Hf\|, \quad \forall f \in \D^\infty(H_0)$
\end{itemize}
then  $\D^\infty(H_0)=\D^\infty(H)$.
\end{cor}
\begin{proof} Condition (i) implies that $B$ leaves $\D^\infty(H_0)$ invariant; therefore $H$ is $t_{H_0}$-continuous (together with all its powers). So it remains to check that $t_H$ is finer than $t_{H_0}$ in order to apply Proposition \ref{CNS}. \\
We will prove by induction that
$$ \|H_0^nf\| \leq \|H^nf\|, \quad \forall f \in \D^\infty(H_0).$$
The case $n=1$ is exactly condition (iii). Now we assume the
statement true for $n-1$. Then we get:
\begin{eqnarray*} \|H_0^n f\|&=& \|H_0(H_0^{n-1} f)\| \\&\leq& \|H (H_0^{n-1} f)\|\\&=&
\|H_0^{n-1}Hf\| \leq \|H^n f\|, \quad \forall f \in
\D^\infty(H_0).
\end{eqnarray*}
since, by (i), $H_0$ and $H$ commute (algebraically) on
$\D^\infty(H_0)$.
\end{proof}

\section{Dynamical aspects}
We now come back to the dynamical problem posed at beginning of
the paper concerning the perturbative situation and again we will
consider the case where $H$ exists. So far, we were able to prove
the convergence of the dynamics only at infinitesimal level
(Corollary \ref{deriv}). The problem of the convergence of
$\,\alpha_L^t (A)$ to $\alpha^t (A)$ is not completely solved
neither in the simple case we are dealing with. Of course, given
$H$ and its spectral projections $Q_L$ as seen in the previous
Section, we can always prove, setting $\hat H_L=Q_L H Q_L$, that
${\rm e}^{i\hat H_Lt}A{\rm e}^{-i\hat H_Lt}$ converges to ${\rm
e}^{iHt}A{\rm e}^{-iHt}$ for any $A$ in $\LD$. What makes here the
difference (and this is the spirit of the whole paper), is that we
are defining the cut-offed hamiltonian $H_L=Q_L^0 H Q_L^0$ via the
spectral projections of the unperturbed hamiltonian $H_0$. This is
of practical interest since  only in very few instances (finite
discrete systems, harmonic oscillators, hydrogen atoms,...) the
spectral projections of $H$ are known. On the other hand, $H_0$
can be chosen with a sufficient freedom to guarantee the knowledge
of the $Q_L^0$.

With this in mind, we consider the problem of finding conditions under which ${\rm e}^{iH_Lt}$ converges, with respect to the topology $\tau_*^{\D}$, to ${\rm e}^{iHt}$.\\
To this aim, we define the operator function
$$ g_L(t)= {\rm e}^{iH_Lt}- {\rm e}^{iHt}= i\int_0^t {\rm e}^{iH_L(t-t')}(H_L-H){\rm e}^{iHt'}dt',$$
the latter equality being got by solving the equation
$$ \dot g_L(t) = i H_L g_L(t) + i(H_L-H){\rm e}^{iHt} $$
which comes directly from the definition of $g_L(t)$. Easy
estimates allow to state the following
\begin{lemma} For each $k \in {\mb N}$ there exists $s \in {\mb N}$ such that
$$\lim_{L\to \infty}\|H_0^{-s}(H_L-H)H_0^k\|=0.$$
\label{59.1}
\end{lemma}
then we have
\begin{prop} If there exists $T>0$ such that, for each $f \in {\mc F}$, $s\in {\mb N} \cup \{0\}$ there exists $M=M({T, f,s})$ such that
 $$ \int_0^t \|f(H_0) {\rm e}^{iH_L(t-t')}H_0^s\|dt' < M, \quad t \in [0,T] $$
then
$$\tau_*^\D-\lim_{L\to \infty}g_L(t)=0.$$
\label{60.1}
\end{prop}
\begin{proof} We have indeed:
\begin{eqnarray*}
\|g_L(t)\|^{f,k} &\leq & \int _0^t \|f(H) {\rm e}^{iH_L(t-t')}(H_L-H){\rm e}^{-iHt'}H^k\|dt'\\
&=&\int _0^t \|f(H) {\rm e}^{iH_L(t-t')}(H_L-H)H^k\|dt'\\
&\leq & C \int_0^t \|f_1(H_0) {\rm e}^{iH_L(t-t')}(H_L-H)H_0^{k_1}\|dt'\\
&\leq & C \int_0^t \|f_1(H_0) {\rm e}^{iH_L(t-t')}H_0^s\|dt' \cdot
\|H_0^{-s}(H_L-H)H_0^{k_1}\|
\end{eqnarray*}
for suitable $C>0$, $f_1 \in {\mc F}$ and $k_1 \in {\mb N}$ and
with $s$ chosen, correspondingly to $k_1$ so that Lemma \ref{59.1}
can be used.
\end{proof}
This proposition implies that the Schr\"odinger dynamics can be
defined. The analysis of the Heisenberg dynamics is more
complicated and will not be considered here.

The assumptions of Proposition \ref{59.1} are indeed quite strong.
They are, of course, verified if the perturbation $B$ commutes
with $H_0$. But this is, clearly, a trivial situation. We will now
discuss a non-trivial example where they are satisfied.

\vspace{2mm} \noindent{\bf Example --} Let $H_0=a^\dagger a$ and
$B=a^n$, $n$ being an integer larger than 1. The conditions on the
domains of the operators discussed in Section 2 are satisfied, as
it is more easily seen working in the configuration space, so that
our procedure can be applied. Here
$Q^0_L=\Pi_0^0+\Pi_1^0+\Pi_2^0+.....+\Pi_L^0$, where $\Pi_i^0$ is
the projection operator of $H_0$,
$H_0=\sum_{l=0}^{\infty}l\Pi_l^0$. Using the algebraic rules
discussed in \cite{bagjmp}, and, in particular the commutation
rules $Q^0_La=aQ_{L+1}^0$ and $\Pi_l^0a=a\Pi_{l+1}^0$, we find
that $H_l=Q^0_LHQ^0_L=HQ^0_L$.

It is a straightforward computation now to check that for any
$f\in {\mc F}$ and for any natural $s$,
$\|f(H)e^{iH_L\tau}H^s\|=\|f(H)(H+(e^{iH\tau}-1)a^nP_{L,n}^0)^s\|$,
where we have defined the following orthogonal projection operator
$$
P_{L,n}^0=\Pi_{L+1}^0+\Pi_{L+2}^0+....+\Pi_{L+n}^0 = Q^0_{L+n}-
Q^0_L.
$$
These seminorm can be estimated for each value of $s$ and it is
not difficult to check that they are bounded by a constant which
depends on $f$, $s$ and $n$ (obviously) but not on $L$ and $\tau$.
Therefore the main hypothesis of Proposition \ref{59.1} is
verified and so the Schr\"odinger dynamics can be defined. We give
the estimate of the above seminorm here only in the easiest non
trivial case, $s=1$. The extension to $s>1$ only increases the
length of the computation but does not affect the result. \beano
\|f(H)\left(H+(e^{iH\tau}-1)a^nP_{L,n}^0\right)\|\leq &&\!\!\!\!\|f(H)H\|+\|f(H)(e^{iH\tau}-1)a^n\|\|P_{L,n}^0\|\\
&&\leq \|f(H)H\|+2\|f(H)a^n\|, \enano which is bounded and
independent of both $L$ and $\tau$.

The same strategy can also be applied to the more general
situation when $B$ is any given polynomial in $a$ and $a^\dagger$.

\vspace{2mm} In order to find more cases in which Proposition
\ref{60.1} can be applied, we begin with the following
\begin{lemma} For each $f \in {\mc F}$ and for each $k, \ell \in {\mb N}$ we have:
$$ \lim_{L\to \infty}\|f(H)(H_L^\ell - H^\ell)H^k\|=0$$
\label{61.1}
\end{lemma}
\begin{proof}We proceed by induction on $\ell$.\\
For $\ell=1$ the statement follows immediately by the equivalence of the topologies and from Lemma \ref{2.2}.\\
Now,
$$\|f(H)(H_L^{\ell+1} - H^{\ell+1})H^k\|\leq \|f(H)H_L(H_L^\ell - H^\ell)H^k\|+ \|f(H)(H_L-H)H^{\ell+k}\|$$
and the second term on the r.h.s. goes to $0$ because we have just proved the induction for $\ell=1$.\\
The first term of the rhs can easily be estimated, making once
more use of the equivalence of the topologies, by a term of the
kind
$$ C' \|f_1(H_0)(H_L^\ell - H^\ell)H^{k_1}\|$$
which goes to zero again because of the induction.
\end{proof}
\begin{prop}If there exists $m \in {\mb N}$ such that $[H_L,H]_{m+1}=0$ then
$$ \int_0^t \|f(H_0) {\rm e}^{iH_L(t-t')}H_0^s\|dt' < \infty, \quad t \in {\mb R}^+$$
for each $f \in {\mc F}$ and for each $s \in {\mb N}\cup \{0\}$
\label{62.2}
\end{prop}
\begin{proof} By the assumption, we have:
\begin{equation} {\rm e}^{iH_L(t-t')} H {\rm e}^{-iH_L(t-t')}= H + i(t-t')[H_L,H] +\cdots \frac{i^m}{m!}[H_L,H]_m.
\label{series}
\end{equation}
Now, using the equivalence between the topologies produced by
$H_0$ and $H$, it is easy to see that:
$$\|f(H_0) {\rm e}^{iH_L(t-t')}H_0^s\|\leq C\|f_1(H)\left( {\rm e}^{iH_L(t-t')} H {\rm e}^{-iH_L(t-t')}\right)^{s_1}\|.$$
Inserting \eqref{series} on the r.h.s. and making use of Lemma
\ref{61.1}, we finally get the estimate:
$$ \|f(H_0) {\rm e}^{iH_L(t-t')}H_0^s\| \leq C\|f_1(H)H^{s_1}\|,$$
and this easily imply the statement.
\end{proof}

Clearly, even if the conditions given above for the
$\tau_*^{\D}$-convergence of ${\rm e}^{iH_Lt}$ to ${\rm e}^{iHt}$
occur, the convergence of $\alpha_L^t (A)$ to $\alpha^t (A)$ is
not guaranteed. For this reason we conclude this Section by
outlining a different possible approach.

Assume that, for each $L \in {\mb R}$ there exists a one-parameter
family  $\beta_L^t (A)$ of linear maps of $\LD$ (not necessarily
an automorphisms group) such that, for each $f \in {\mc F}$ and $k
\in {\mb N}$,
\begin{equation} \|f(H)\left(\beta_L^t(A)-\alpha_L^t(A)\right)H^k\|\to 0 \quad \mbox{as }L \to \infty
\label{qd}
\end{equation}
Clearly the convergence of $\beta_L^t (A)$ to $\alpha^t(A)$ would directly lead to the solution of the dynamical problem. We want to stress that $\beta_L^t$ could be rather {\em unusual} and, therefore, it should be only considered as a technical tool.\\
In general, however, the possibility of finding a good definition
for the $\beta^t_L$'s that allows \eqref{qd} to hold, is quite
difficult and the lesson of the previous discussion on the
convergence of ${\rm e}^{iH_Lt}$ is that strong assumptions must
be imposed in order to get it.\\ A weaker condition on the
$\beta^t_L$'s, whose content of information is, nevertheless,
non-empty, would consists in requiring, instead of \eqref{qd},
that
\begin{itemize}
\item[(a)] $\beta_L^t(A)$ converges to $\alpha^t(A)$, for each $A \in \LD$ together with all time derivatives. This means that an Heisenberg dynamics $e^{iHt}(\cdot )e^{-iHt}$ can be recovered.
\item[(b)] As for the Schr\"odinger dynamics, that is for the dynamics in the space of vectors, we ask $\beta_L^t$ of being in same way (to be specified further) {\em generated} by a family of bounded operators which, when applied to any $\Psi\in\D$, is $t_H$-convergent together with all time derivatives.
\end{itemize}
This happens, for instance, in the case where $H$ exists, if we
define
$$ \beta_L^t (A)= V_L^t A V_L^{-t} .$$
where $ V_L^t:= Q^0_L {\rm e}^{iHt} Q^0_L $ and  the $Q^0_L$'s are the projection of $H_0$. \eqref{proj}.\\
Under these assumptions, $V_L$ is a well-defined bounded operator
of $\LD$, and the following Proposition holds:
\begin{prop}In the above conditions, the following statements hold:
\begin{itemize}
\item[(i)] $t_H-\lim_{L\to \infty} V_L^t \psi = \psi(t):={\rm e}^{iHt} \psi, \quad \forall \psi \in \D$
\item[(ii)]$\tau_*^{\D}-\lim _{L\to \infty} V_L^t = {\rm e}^{iHt}$
\item[(iii)]$\tau_*^{\D}-\lim _{L\to \infty} \beta_L^t (A)= \alpha^t (A):={\rm e}^{iHt}A{\rm e}^{-iHt}, \quad \forall A \in \LD$
\end{itemize}
and, more generally:
\begin{itemize}
\item[(i')] $t_H-\lim_{L\to \infty}\frac{d^n}{dt^n} V_L^t \psi =\frac{d^n}{dt^n} \psi(t), \quad \forall \psi \in \D, \; n \in {\mb N}\cup\{0\}$
\item[(iii')] $\tau_*^{\D}-\lim _{L\to \infty} \frac{d^n}{dt^n}\beta_L^t (A)= \frac{d^n}{dt^n}\alpha^t (A), \quad \forall A \in \LD, \; n \in {\mb N}\cup\{0\}$
\end{itemize}
\label{49.1}
\end{prop}

The proof of this Proposition follows from the equivalence between the topologies generated by $H$ and $H_0$, proved in Lemma 2.1.\\
This approach, which is only one of the possible strategies when
$H$ exists, could be of a certain interest for situations when the
dynamics can only be obtained via a net of operators
$H_L=H_0+B_L$, $H_0$ being the free Hamiltonian and $B_L$ being a
regularized perturbation.  In this case the approach to the
thermodynamical limit could involve the family of bounded
operators $ V_{L,M}^t:= Q^0_L {\rm e}^{iH_Mt} Q^0_L $, and one can
try to extend the above results. A further analysis on this
subject is currently {\em work in progress}.

 \section{Outcome and possible developments}

In this paper we have analyzed a possible approach to define an
algebraic dynamics when a free hamiltonian $H_0$ is perturbed by
an operator $B$ which essentially leaves the domain of all the
powers of $H_0$ invariant.

What is still missing, how we discussed in the Introduction, is
the analysis of the situation where the definition of the dynamics
is not straightforward since it should follow from a net of
operators $\{H_L\}$ whose limit does not exist in any physical
topology \cite{thi, bagarello}. In this case a possible approach
can be made in terms of derivations, for instance, in the way
explained below.

Let us suppose that to a free spatial derivation
$\delta_0(.)=i[H_0,.]$ a perturbation term $\delta_P$ is added, so
that
$$
\delta(A)= \delta_0(A)+\delta_P(A),\quad A\in \LD.
$$
In this case we define $\eta_L(A)=Q^0_L\delta(A)Q^0_L$, with $A\in
\LD$ and $Q^0_L$ as in the previous Sections. It is easy to see
that $\eta_L$ is not in general a derivation because the Leibniz
rule may fail. Let $\Delta_L$ be a map on $\LD$ which has the
property that $\delta_L=\eta_L+\Delta_L$ satisfies the Leibniz
rule together with the other properties of a derivation. Of
course, this map is not unique since, for instance, we can always
add a commutator $i[H',.]$ to $\Delta_L$, with any self-adjoint
operator $H'$, without affecting the properties of a derivation
(we should only care about domain problems in choosing $H'$!).
>From a physical point of view it is reasonable to expect that
$\Delta_L$ can be chosen in such a way that
$\|\Delta_L(A)\|^{f,k}\rightarrow 0$ with $L$ since we would like
to recover the original derivation $\delta$ after removing of the
cutoff and we know from Corollary \ref{deriv} that
$\|(\eta_L(A)-\delta(A))\|^{f,k}\rightarrow 0$. If also $\delta_P$
is spatial, then is not difficult to give an explicit expression
for $\Delta_L$ and to check that the requirements above are
satisfied. In this case in fact
$$
 \Delta_L(A)=\{Q^0_LH_0,[Q^0_L,A]\}+Q^0_LB[Q^0_L,A]+[Q^0_L,A]BQ^0_L,
$$
where $\{X,Y\}=XY+YX$.

Once we have introduced $\delta_L$ the next step is to find
conditions for this map to be spatial. The related operator $H_L$,
which we expect to be of the form $Q^0_L(H_0+B)Q^0_L$ for a
suitable self-adjoint operator $B$, can be used to perform the
same analysis as that discussed in the previous Section.

Of course this is by no means the only possibility of approaching
this problem, but is the one which is closer to our previous
analysis, and in this perspective, is particularly relevant for
us. We hope to discuss this problem in full details in a future
paper.

\vspace{6mm}

\noindent{\large \bf Acknowledgments} \vspace{5mm}

    F.B. acknowledges financial support by the Murst, within the  project {\em Problemi
Matematici Non Lineari di Propagazione e Stabilit\`a nei Modelli
del Continuo}, coordinated by Prof. T. Ruggeri.

\end{document}